\documentclass[11pt]{article}   	
\usepackage{geometry}                		
\usepackage{graphicx}				
\usepackage{amssymb}
\usepackage{amsthm}
\usepackage{enumerate}
\usepackage{bm}
\usepackage{extarrows}
\usepackage{natbib}
\usepackage{hyperref}
\usepackage{booktabs}
\usepackage{xcolor}
\usepackage{pdflscape}

\newcommand{\E}[1]{\mathbb{E}\left[{#1}\right]}
\newcommand{\Ex}[2]{\mathbb{E}_{#1}\left[{#2}\right]}
\newcommand{\pde}[2]{\frac{\partial{#1}}{\partial{#2}}}
\newcommand{\norm}[1]{\left\| #1 \right\|}
\newcommand{\inner}[2]{\left\langle #1, #2 \right\rangle}

\newcommand{\K}{\mathbf{K}}

\newcommand{\bmat}[1] {\begin{bmatrix} #1 \end{bmatrix}}
\newcommand{\Op}[1]{O_p\left(#1\right)}
\newcommand{\op}[1]{o_p\left(#1\right)}
\newcommand{\I}[1]{I\left(#1\right)}

\newtheorem{theorem}{Theorem}[section]
\newtheorem{lemma}{Lemma}[section]
\newtheorem{corollary}{Corollary}[theorem]

\newcommand{\nblue}[1]{{#1}}

\title{Tree Boosted Varying Coefficient Models}
\author{Yichen Zhou, Giles Hooker\\ \\ Department of Statistics and Data Science, Cornell University}

\begin{document}
\maketitle
\begin{abstract}
This paper investigates the integration of gradient boosted decision trees and varying coefficient models. We introduce the tree boosted varying coefficient framework which  justifies the implementation of decision tree boosting as the nonparametric effect modifiers in varying coefficient models. This framework requires no structural assumptions in the space containing the varying coefficient covariates, is easy to implement, and keeps a balance between model complexity and interpretability. To provide statistical guarantees, we prove the asymptotic consistency of the proposed method under the regression settings with $L^2$ loss. We further conduct a thorough empirical study to show that the proposed method is capable of providing accurate predictions as well as intelligible visual explanations.
\end{abstract}

\section{Introduction}
In this paper we study the amalgamation of gradient boosting, especially gradient boosted decision trees \citep[GBDT or GBM:][]{friedman2001greedy}, and varying coefficient models \citep[VCM:][]{hastie1993varying}. A varying coefficient model is a semi-parametric model with coefficients that change along with each input. Under a general statistical learning setting with a set of covariates and some response of interest, a VCM isolates part of those covariates as \textit{effect modifiers} based on which model coefficients are determined through a few varying coefficient mappings. These coefficients then get joined with the remaining covariates to generate a parametric prediction. To elaborate, consider performing least square regression on $(X, Z, Y) \in \mathbb{R}^p \times \mathbb{A} \times \mathbb{R}, i=1,\dots, n$ where $X = (X^1, \dots, X^p)$, $X$ and $Z$ are the covariates and $Y$ the response. One VCM regression can take the form of 
\begin{align}
\label{eqn:vcm1}
g(\E{Y|X, Z}) = \beta^{0}(Z) + \sum_{i=1}^p \beta^i(Z) X^i,
\end{align}
with the parametric part being a generalized linear model with the link function $g$. In this context we would like to refer to $X$ as the \textit{predictive covariates} and $Z$ the \textit{action  covariates (effect modifiers)} which are drawn from $\mathbb{A}$ the \textit{action space}. $\beta^i(\cdot): \mathbb{A} \to \mathbb{R}, i=0,1,\dots,p$ are, conventionally nonparametric, \textit{varying coefficient mappings}. While (\ref{eqn:vcm1}) maintains the linear structure, due to the dependence of $\beta$ on any given $Z$, the model belongs to a more complicated and flexible model space rather than the corresponding generalized linear model. 

\nblue{
Our proposed model, tree boosted VCM, utilizes ensembles of gradient boosted decision trees as the varying coefficient mappings $\beta$. To demonstrate, for each $\beta^i, i=0,\dots, p$, let
$$
\beta^i(z) = \sum_{j=1}^b t^i_j(z),
$$ 
an additive boosted tree ensemble of size $b$ with each $t_j^i$ a decision tree constructed sequentially through gradient boosting. We will postpone the details of model construction to Section 2. This strategy yields a model of
\begin{align}
\label{eqn:treeboostedvmc}
g(\E{Y|X, Z}) = \sum_{j=1}^b t^0_j(Z) + \sum_{i=1}^p \left(\sum_{j=1}^b t^i_j (Z)\right) X^i.
\end{align}
}

Introducing VCM aligns with our attempt to answer the rising concern about model intelligibility and transparency, around which there are two branches of methods. We can either apply \textit{post hoc} methods such that state-of-the-art ``black box" models are constructed before we grant them meanings through analyzing their results. There is a sizable literature on this topic, from the appearance of local methods \citep{ribeiro2016should} to recent applications on neural nets \citep{zhang2018visual}, random forests \citep{mentch2016quantifying, basu2018iterative} and complex model distillation \citep{lou2012intelligible, lou2013accurate, tan2017distill}. However, objectivity is one inevitable challenge of tying explanations to models, especially in the presence of plentiful universal local methods capable of dealing with most models. Any use of \textit{post hoc} analysis may be subject to justify the chosen explanatory method over the others, which is likely to add an additional explanation selection phase on top of the existing model selection. 

On the other hand, another branch of methods attempts to build interpretability into model structures, meaning that models should be the integration of simple and intelligible building blocks that they become accountable by human inspection once trained. Examples of this range from simple models as generalized linear models and decision trees, to models that guarantee monotonicity \citep{you2017deep, chipman2016high} or have identifiable components \citep{melis2018towards}. Although having the advantage of not requiring \textit{post hoc} examination, in contrast to the aforementioned methods, self-explanatory models are restricted by their possible model complexity and flexibility, potentially limiting their accuracy. This lack of flexibility also implies that such a model, unless possessing a granular structure, may only provide global interpretation because all observations are reasoned via an identical procedure. Such behavior prevents us from zooming into a small region in the sample space.

Following this discussion, VCM belongs to the second category as long as the involved parametric models are intelligible. It is an instant generalization of parametric methods to allow the use of local coefficients, which leads to improvements in model complexity and accuracy, whereas the predictions are still produced through parametric relations between predictive covariates and coefficients. This combination demonstrates a feasible means to balance the trade-off between flexibility and intelligibility. 

A great amount of research has been conducted to study the asymptotic properties of different VCMs when splines or kernel smoothers are implemented as the nonparametric varying coefficient mappings. We refer the readers to \citet{park2015varying} for a comprehensive review. In this paper we intend to conclude similar results regarding the asymptotics of tree boosted VCM.

\subsection{Models under VCM} Under the settings of (\ref{eqn:vcm1}), \citet{hastie1993varying} pointed out that VCM is the generalization of generalized linear models, generalized additive models, and various other semi-parametric models with careful choices of the varying coefficient mapping $\beta$. 

We would like to mention two special cases that have drawn our attention. One is the functional trees introduced in \citet{gama2004functional}. A functional tree segments the action space into disjoint regions, after which a parametric model gets fitted within each region using sample points inside. Logistic regression trees, for which there is a sophisticated building algorithm \citep[LOTUS:][]{chan2004lotus}, belong to such model family. Their prediction on $(x_0, z_0)$ is
$$
P(\hat{y}_0=1) =  \sum_{i=1}^K \frac{1}{1+e^{-x_0^T\beta_i}}\cdot \I{z_0 \in A_i} = \frac{1}{1+e^{-x_0^T\beta_k}},
$$
provided $\mathbb{A} = \coprod_{i=1}^K A_i$ the tree segmentation, $z_0 \in A_k$ and $\beta_k = \beta(z) ,\forall z \in A_k$. The conventional approach to determine functional tree structure is to recursively enumerate through candidate splits and choose the one that reduces the training loss the most between before and after splitting. Despite of the guaranteed stepwise improvement, such greedy strategy has the side effect of being both time consuming and mathematically intractable. 

Another case is the partially linear regression that assumes
$$
Y = X^T\beta + f(Z) + \epsilon_Z, \quad \epsilon_Z \sim N(0, \sigma^2(Z)), 
$$
where $\beta$ is a global linear coefficient \citep[see][]{hardle2012partially}. It is equivalent to a least square VCM with all varying coefficient mappings except the intercept being constant. 

\subsection{Trees and VCM}
While popular choices of varying coefficient mappings are either splines or kernel smoothers, it is a natural transition to consider exercising decision trees \citep[CART:][]{breiman1984classification} and decision tree ensembles to serve as these nonparametric mappings. Using trees enables us to work adaptively with any action space $\mathbb{A}$ compatible with decision tree splitting logic, for example an arbitrary high dimensional mixture of continuous and discrete quantities, whereas traditional methods require to craft model structures case by case depending on the given $\mathbb{A}$. 

We start with the straightforward attempts to utilize a single decision tree as varying coefficient mappings \citep{buergin2017coefficient, berger2017tree}. Although having a simple form, these implementations are also subject to the instability caused by the greedy tree building algorithm. Moreover, the mathematical intractability of decision trees prevents these single-tree based varying coefficient mappings from provable optimality. This instead suggests implementations through tree ensembles of either random forests or gradient boosting. \nblue{One example is to use the linear local forests introduced in \citet{friedberg2018local} that perform local linear regression with an honest random forest kernel, while the predictive covariates $X$ are reused as the action covariates $Z$. In terms of boosting methods, \citet{wang2014boosted} proposed the first tree boosted VCM algorithm. They reduced the empirical risk by boosting using functional trees to fit the residuals to improve model coefficients, resulting in models of
$$
g(\E{Y|X, Z}) = \left(\sum_{j=1}^b t_j(Z)\right)^T (1, X),
$$ where each $t_j$ returns a $(p+1)$ dimensional response.} However, building a functional tree ensemble requires the construction and comparison of massive amounts of submodels and the joint optimization of all coefficients. In contrast, we aim to perform gradient boosting down on the coefficient level to comply with the standard boosting framework in order to separate the coefficients and to make tree boosted VCM coherent with existing boosting theories. 

In the following sections, we explore the feasibility and statistical properties of adopting generic gradient boosted decision trees to serve as the nonparametric varying coefficient mappings for VCM. In Section 2, we share the perspective of analyzing such models as \textit{local gradient descent} which creates functional coefficients and optimizes using local information. We will prove the consistency of this method in Section 3 and present a few empirical study results in Section 4. Further discussions on potential variations of this method follow in Section 5. 

\section{Tree Boosted Varying Coefficient Models}
\subsection{Notations}
We will use the following notations in our discussion. We use superscripts $0, \dots, p$ to indicate individual components, i.e. $\beta = (\beta^0, \dots, \beta^p)^T$, and subscripts $1, \dots, n$ to indicate sample points or boosting iterations. For any $X$ when there is no ambiguity we assume $X$ contains the intercept column, i.e. $X = (1, X^1,\dots, X^p)$, so that $X^T\beta = \beta^0 + \sum_{i=1}^p X^i \beta^i$ can be used to specify a linear regression. 
\subsection{Boosting Framework}
We start by looking at a parametric generalized linear model with coefficients $\beta \in \mathbb{R}^{p+1}$ using gradient descent. Given sample $(x_1, z_1, y_1),\dots, (x_n, z_n, y_n)$ and a loss function $l$, gradient descent minimizes the empirical risk to search for the optimal $\hat{\beta}^*$ as
$$
L(\hat{\beta}) =  \frac{1}{n}\sum_{i=1}^n l(y_i, x_i^T\hat{\beta}), \quad \hat{\beta}^* = \arg \min_{\hat{\beta}} L(\hat{\beta}).
$$
To improve an interim $\hat{\beta}$, we move it in the negative gradient direction 
$$
\Delta_{\hat{\beta}} = \nabla_{\beta} L = -\nabla_{\beta} \left(\frac{1}{n}\sum_{i=1}^n l(y_i, x_i^T\beta) \Big|_{\beta = \hat{\beta}}\right),
$$
to obtain a new iteration $\hat{\beta}' = \hat{\beta} + \lambda \Delta_{\hat{\beta}}$ for a positive and small learning rate $\lambda \ll 1$. 

In order to extend this setting to varying coefficient models, we instead consider $\beta$ to be a mapping $\beta = \beta(z): \mathbb{A} \to \mathbb{R}^{p+1}$ so that it will apply to the covariates based on their values in the action space. Writing estimate of $\beta$ by $\hat{\beta}: \mathbb{A} \to \mathbb{R}^{p+1}$, the empirical risk remains a similar form
$$
L(\hat{\beta}) = \frac{1}{n}\sum_{i=1}^n l(y_i, x_i^T\hat{\beta}(z)).
$$
We perform the same gradient calculation as above, but only pointwisely for now. It produces the negative gradient direction at the point $z_i$
\begin{align}
\label{fml:gradient}
\Delta_{\beta}(z_i) = - \nabla_{\beta}l(y_i, x_i^T\beta) \Big|_{\beta = \hat{\beta} (z_i)}.
\end{align}
As a result, we get the functional improvement of $\hat{\beta}$ captured at each of the sample points, i.e. $(z_1, \Delta_{\beta}(z_1)), \dots, (z_n, \Delta_{\beta}(z_n))$. This observation leads us to employ gradient descent in functional space, also known as boosting \citep{friedman2001greedy}. For any function family $\mathcal{T}$ capable of regressing $\Delta_{\beta}(z_1), \dots, \Delta_{\beta}(z_n)$ on $z_1,\dots, z_n$, the corresponding ordinary boosting framework works as follows.
\begin{enumerate}[(B1)]
\item Start with an initial guess of $\hat{\beta}_0(\cdot)$. 
\item For each component $j = 0,\dots, p$ of $\hat{\beta}_b, b\geq 0$, we calculate the pseudo gradient at each point as
\begin{gather*}
\Delta^j_{\beta_i} = -  \pde{l(y_i, x_i^T\beta)}{\beta^j}\Big|_{\beta = \hat{\beta}_b(z_i)},
\end{gather*}
for $i=1,\dots, n$.
\item For each $j$, find a good fit $t_{b+1}^j \in \mathcal{T}: \mathbb{A} \to \mathbb{R}$ on $(z_i, \Delta^j_{\beta_i}), i=1,\dots, n$.
\item Update $\hat{\beta}_b$ with learning rate $\lambda \ll 1$.
\begin{gather*}
\hat{\beta}_{b+1}(\cdot) = \hat{\beta}_{b}(\cdot) + \lambda \begin{bmatrix}t_{b+1}^0(\cdot) \\ \vdots \\ t_{b+1}^p(\cdot)\end{bmatrix}. 
\end{gather*}
\end{enumerate}

When the spline method is implemented, $\mathcal{T}$ is closed under addition so that we will expect the result of (B4) to be expressed as a set of coefficients of basis functions for $\mathcal{T}$. On the other hand, when we apply decision trees in place of (B3):
\begin{enumerate}[(B3')]
\item For each $j$, build a decision tree $t_{b+1}^j: \mathbb{A} \to \mathbb{R}$ on $(z_i, \Delta^j_{\beta_i}), i=1,\dots, n$,
\end{enumerate}
the resulting varying coefficient mapping will be an additive tree ensemble, whose model space varies based on the ensemble size. We will refer to this method as \textit{tree boosted VCM}. 

Notice that the strategy of building a decision tree in (B3') influences the properties of the obtained tree boosted VCM. The standard CART strategy executes as follows.
\begin{enumerate}[(D1)]
\item Start at the root node. 
\item Given a node, numerate candidate splits and evaluate them using all $(z_i, \Delta^j_{\beta_i})$ such that $z_i$ is contained in the node.
\item Split on the best candidate split.
\item Keep splitting until stopping rules are met to form terminal nodes.
\item Calculate fitted terminal values in each terminal node using all $(z_i, \Delta^j_{\beta_i})$ such that $z_i$ is contained in the terminal node. 
\end{enumerate}
 
Recent developments on decision trees also suggest alternative strategies that produce better theoretical guarantees. We may consider \textit{subsampling} that generates a subset $w \subset \{1,\dots,n\}$ and only uses sample points indexed by $w$ in (D2). 
\begin{enumerate}[(D2')]
\item Given a node, numerate candidate splits and evaluate them using all $(z_i, \Delta^j_{\beta_i})$ such that $i \in w$ and $z_i$ is contained in the node.
\end{enumerate}
We may also consider \textit{honesty} \citep{wager2018estimation} which avoids using the responses, in our case $\Delta^j_{\beta_i}$, twice during both deciding the tree structure and deciding terminal values. For instance, a version of completely random trees discussed in \citet{zhou2018boulevard} chooses the splits using solely $z_i$ without evaluating the splits by the responses $\Delta^j_{\beta_i}$ in place of steps (D2) and (D3). 
\begin{enumerate}[(D2*)]
\item Given a node, choose a random split based on $z_i$'s contained in the node.
\end{enumerate}

\subsection{Local Gradient Descent with Tree Kernels} 
Decision tree fits in (B3') generate local linear combinations of pseudo-gradients thanks to the grouping effect carried by decision tree terminal nodes. To elaborate from a generic viewpoint, \nblue{for all tree building strategy we discussed above} we can introduce a kernel smoother $K: \mathbb{A}\times\mathbb{A} \to \mathbb{R}$ such that the estimated gradient at any new $z$ is given by 
\begin{align}
\label{fml:lgd}
\Delta_{\beta}(z) = - \sum_{i=1}^n \left(\nabla_{\beta}l(y_i, x_i^T\beta) \Big|_{\beta = \hat{\beta}(z_i)}\right) \cdot \frac{K(z, z_i)}{\sum_{j=1}^n K(z, z_j)}.
\end{align}
\nblue{In other words, with a fast decaying $K$, (\ref{fml:lgd}) can estimate the gradient at $z$ locally using weights given by
$$
S(z, z_i) = \frac{K(z, z_i)}{\sum_{j=1}^n K(z, z_j)}.
$$
We would like to define such method as \textit{local gradient descent}. 

During standard tree boosting employing CART strategy, after a decision tree is constructed each iteration, its induced smoother $K$ assigns equal weights to all sample points in the same terminal node. If we write $A(z_i) \subset \mathbb{A}$ the region in the action space corresponding to the terminal node containing $z_i$, we have $K(z, z_i) = I(z \in A(z_i))$ and we define the following 
$$  \K(z, z_i) \triangleq S(z, z_i)=  \frac{\I{z \in A(z_i)}}{\sum_{j=1}^n I(z_j \in A(z_i))}$$
to be the \textit{tree structure function} where we also use the convention that $0/0=0$. The denominator is the size of $z_i$'s terminal node and is equal to $
\sum_{j=1}^n I(z_j \in A(z))$ when $z$ and $z_i$ fall in the same terminal node. In the cases where subsampling or completely random trees are employed for the purpose of variance reduction, $\K$ will be taken to be the expectation such that
$$
\K(z, z_i) \triangleq \E{S(z, z_i)} = \E{\frac{I(z \in A(z_i))I(i \in w)}{\sum_{j=1}^n I(z_j \in A(z_i))I(i \in w)I(j \in w)}}.
$$
whose properties have been studied by \citet{zhou2018boulevard}. This expectation is taken over all possible tree structures and, if subsampling is applied, all possible subsamples $w$ of a fixed size, and the denominator in the expectation is again the size of $z_i$'s terminal node. 

In particular, by carefully choosing the rates for tree construction, this tree structure function is related to the random forest kernel introduced in \citet{scornet2016random} that takes the expectation of the numerator and the denominator separately as
$$
\K_{RF}(z, z_i) = \frac{\E{\I{z \in A(z_i)}}}{\E{\sum_{j=1}^n \I{z_j \in A(z_i)}}},
$$
in the sense that the deviations from these expectations are mutually bounded by constants.}

Gradient boosting applied under nonparametric regression setting has to be accompanied by regularization such as using a complexity penalty or early stopping to prevent overfitting. When decision trees are implemented as the base learners, the complexity penalty is implicitly embedded in the tree parameters such as tree depth and terminal node size, while early stopping can be enforced during training. In fact, while we keep the parametric linear structure in VCM, local neighborhood weighting used for fitting the nonparametric coefficient mappings still adds to the model complexity. Therefore moderate restrictions, especially growth rates, have to be applied to avoid building saturated models with respect to the action space.

\subsection{Examples}
\nblue{Tree boosted VCM generates a two-phase model such that the varying coefficient mappings generate effect modifiers and these effect modifiers join with predictive covariates linearly. In order to understand the varying coefficient mappings on the actions space, we provide a visualized example here by considering} the following data generating process:
\begin{gather*}
X \sim \mbox{Unif}[0,1]^3, Z = (Z^1, Z^2) \sim \mbox{Unif}[0,1]^2, \epsilon \sim N(0, 0.25),\\
Y = X^T \bmat{1\\3\\-5}\cdot I(Z^1+Z^2 < 1) + X^T \bmat{0\\10\\0} \cdot I(Z^1+Z^2 \geq 1) + \epsilon.
\end{gather*}
\nblue{We generate a sample of size 1,000 from the above distribution, apply the tree boosted VCM with 400 trees, and obtain the following estimation of the varying coefficient mappings $\beta$ on $Z$ in Figure \ref{fig:vcm_1}. Our fitted values accurately capture the true coefficients. }
\begin{figure}[h] 
   \includegraphics[width=0.325\linewidth]{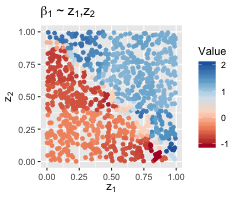} 
   \includegraphics[width=0.325\linewidth]{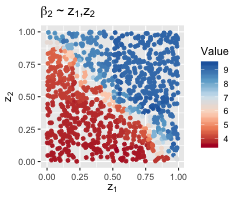} 
   \includegraphics[width=0.325\linewidth]{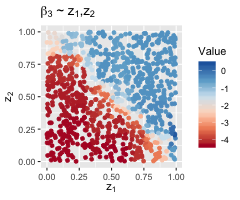} 
   \caption{Example of varying coefficient mappings on the action space under the OLS settings.}
   \label{fig:vcm_1}
\end{figure}

Switching to logistic regression setting and assuming similarly that
$$
\mbox{logit} P(Y=1) = \exp \left(X^T \bmat{1\\3\\-5}\cdot I(Z^1+Z^2 < 1) + X^T \bmat{0\\10\\0} \cdot I(Z^1+Z^2 \geq 1)\right),
$$
with a sample of size 1,000, Figure \ref{fig:vcm_2} presents equivalent plots for our tree boosted VCM. These results are less clear since logistic regression produces more volatile gradients. 
\begin{figure}[h] 
   \includegraphics[width=0.32\linewidth]{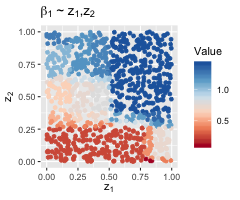} 
   \includegraphics[width=0.32\linewidth]{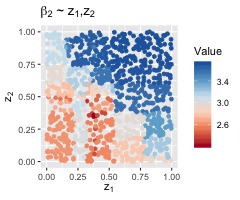}    
   \includegraphics[width=0.32\linewidth]{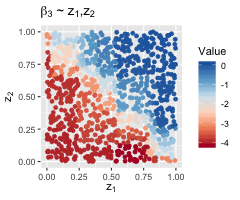} 
   \caption{Example of varying coefficient mappings on the action space under the logistic regression settings.}
   \label{fig:vcm_2}
\end{figure}

In both cases our methods correctly identify $\beta(z)$ as segmenting along the diagonal in $z$, providing clear visual identification of the behavior of $\beta(z)$. These figures are evidence of the capability of tree boosted VCM to find the varying coefficients without posting structural assumptions on the action space. Further empirical studies are presented in Section 4.

\section{Tree Boosted VCM Asymptotics}
There is a large literature providing statistical guarantees and asymptotic analyses of different versions of VCM with varying coefficient mappings obtained via splines or local smoothers \citep{park2015varying, fan1999statistical, fan2005profile}. In this section we will demonstrate the asymptotic analyses of tree boosted VCM under mild conditions.

\subsection{Tree Boosted VCM with $L^2$ Loss} Consider $L^2$ boosting setting for regression. Given the relationship  
$$
Y = f(X,Z) + \epsilon_Z, \quad \epsilon_Z \sim N(0, \sigma^2_Z),
$$
we work with the following assumptions.
\begin{enumerate}[(E1)]
\item Unit support of $X$ that $\mbox{supp } X = \{1\}\times[-1,1]^{p}$, which is achievable by standardizing without loss of generality for any finitely supported $X$.
\item Uniform bounded noise variance that $\sigma_Z \leq \sigma^*$.
\item $L^2$ loss that $L(u,y) = \frac{1}{2} (u-y)^2.$
\end{enumerate}
Under these conditions, evaluating the pseudo-gradient given in (\ref{fml:gradient}) yields
$$
\Delta_{\beta}(z_i) = - \nabla_{\beta}l(y_i, x_i^T\beta) \Big|_{\beta = \hat{\beta} (z_i)} = (y_i -x_i^T\hat{\beta}(z_i))\cdot x_i.
$$
For an existing terminal node $R \subseteq \mathbb{A}$, as per (\ref{fml:lgd}), the decision tree update in $R$ is
\begin{align}
\label{fml:treeupdate}
\Delta_{\beta}(z \in R) = \frac{\sum_{i=1}^n(y_i -x_i^T\hat{\beta}(z_i))\cdot x_i\cdot \I{z_i \in R}}{\sum_{i=1}^n \I{z_i \in R}},
\end{align}
\nblue{
and should subsample $w$ be present
\begin{align*}
\Delta_{\beta}(z \in R; w) = \frac{\sum_{i=1}^n(y_i -x_i^T\hat{\beta}(z_i))\cdot x_i\cdot \I{z_i \in R}\I{i \in w}}{\sum_{i=1}^n \I{z_i \in R}\I{i \in w}}.
\end{align*}}

\subsection{Decomposing Decision Trees}
We assume the action space $\mathbb{A}$ involves only continuous and categorical covariates, therefore we will consider its embedding into a Euclidean space $\mathbb{R}^d$ where $d = \mbox{dim}(\mathbb{A})$ is the  dimension of the embedding. Denote 
$
\mathcal{R} = \{(a_1, b_1] \times\cdots\times (a_d, b_d] | -\infty \leq a_i \leq b_i \leq \infty\}
$
 the collection of all hyper rectangles in $\mathbb{A}$. This set includes all possible terminal nodes of any decision tree built on $\mathbb{A}$. Given the distribution $(Z, 1, X) \sim \mathbb{P}$, we define the inner product 
$\inner{f_1}{f_2} = \Ex{\mathbb{P}}{f_1f_2},$
and the norm $\norm{\cdot} = \norm{\cdot}_{\mathbb{P},2}$ on the sample space $\mathbb{A} \times\{1\}\times [-1,1]^p$. For a sample of size $n$, we write the (unscaled) empirical counterpart by $\inner{f_1}{f_2}_n = \sum_{i=1}^n f_1(x_i, z_i) f_2(x_i, z_i),$ 
such that $n^{-1}\inner{f_1}{f_2}_n \to \inner{f_1}{f_2}$ by the law of large numbers, with a corresponding norm $\norm{\cdot}_n$.

Consider the following classes of functions on $\mathbb{A} \times \{1\}\times[-1, 1]^{p}$.
\begin{itemize}
\item $\mathcal{H} = \left\{h_R(x,z) = \I{z \in R} | R \in \mathcal{R} \right\}$, indicators of hyper rectangles.
\item $\mathcal{G} = \left\{g_{R, j}(x,z) =  \I{z \in R} \cdot x^j | R \in \mathcal{R}, j=0,\dots, p \right\}$ constants and coordinate mappings in hyper rectangles in $\mathcal{R}$. In particular we write  $1 = x^0$ so that $g_{R, 0} = h_R$.
\end{itemize}

\citet{buhlmann2002consistency} established a consistency guarantee for tree-type basis functions for $L^2$ boosting, in which the key point is to bound the gap between the boosting procedure and its population version by the uniform convergence in distribution of the family of indicators for hyper rectangles. We take a similar approach, for which we have to extend the uniform convergence to a broader function class defined as $\mathcal{G}$ as defined above. The following lemma provides uniform bounds on the asymptotic variability pertaining to $\mathcal{G}$ using Donsker's theorem \citep[see][]{van1997weak}.

\begin{lemma}
\label{lemma:donsker}
For given $L^2$ function $f$ and random sub-Gaussian noise $\epsilon$, the following empirical gaps  
\begin{enumerate}
\item $\xi_{n,1}  = \sup_{R \in \mathcal{R}}  \left| \norm{h_R}^2 - \frac{1}{n}\norm{h_R}_n^2  \right|, $
\item $\xi_{n,2}  = \sup_{R \in \mathcal{R}, j=0, \dots, p}  \left| \norm{g_{R, j}}^2 - \frac{1}{n}\norm{g_{R, j}}_n^2 \right|,$  
\item $\xi_{n,3}  = \sup_{R \in \mathcal{R}, j=0, \dots, p}  \left| \inner{f}{g_{R, j}} - \frac{1}{n}\inner{f}{g_{R,j}}_n  \right|, $
\item $\xi_{n,4}  = \sup_{R \in \mathcal{R}, j=0, \dots, p}  \left|\frac{1}{n}\inner{\epsilon}{g_{R, j}}_n  \right|,  $
\item $\xi_{n,5}  = \sup_{R_1, R_2 \in \mathcal{R}, j, k=0, \dots, p}  \left| \inner{g_{R_1, j}}{g_{R_2, k}} - \frac{1}{n}\inner{g_{R_1, j}}{g_{R_2, k}}_n  \right|, $
\item $\xi_{n,6}  = \left| \frac{1}{n}\norm{f+\epsilon}_n^2 - \norm{f+\epsilon}^2\right|,$
\end{enumerate}
satisfy that $\xi_n = \max_{i=1}^6 \xi_{n,i} = \Op{n^{-\frac{1}{2}}}.$
\end{lemma}

Introduce the empirical remainder function $\hat{r}_b$ such that 
$$
\hat{r}_0(x,z) = f(x,z) + \epsilon,\quad \hat{r}_b(x,z) = f(x,z) + \epsilon - \hat{\beta}_b(z)^Tx, b>0,
$$
i.e. the remainder term after $b$-th boosting iteration. Further, consider the $b$-th iteration utilizing $p+1$ decision trees whose disjoint terminal nodes are $R^j_1,\dots, R^j_m \in \mathcal{R}$ for $j=0,\dots,p$ respectively. (\ref{fml:treeupdate}) is equivalent to the following expression for the boosting update of the remainder $\hat{r}_b$ 
\begin{align}
\label{eqn:improve}
\hat{r}_{b+1} = \hat{r}_b - \lambda \sum_{i=1}^m \sum_{j=0}^p \frac{n^{-1}\inner{\hat{r}_b}{g_{{R^j_i}, j}}_n}{n^{-1}\norm{h_{R^j_i}}_n^2} g_{{R^j_i}, j}, 
\end{align}
or, for simplicity, we flatten the subscripts when there is no ambiguity such that
$$
\hat{r}_{b+1} = \hat{r}_b - \lambda \sum_{i=1}^{m(p+1)} \frac{n^{-1}\inner{\hat{r}_b}{g_{b,i}}_n}{n^{-1}\norm{h_{b,i}}_n^2} g_{b,i},
$$
\nblue{
where as defined above, $g_{b,i} = g_{{R^j_i}, j} = \I{z \in R^j_i}\cdot x^j$. }
Although the update involves $m(p+1)$ terms, only $p+1$ of them are applicable for a given $(x,z)$ pair as the result of using disjoint terminal nodes. 

Further, \citet{mallat1993matching} and \citet{buhlmann2002consistency} suggested that we consider the population counterparts of these processes defined by the remainder functions starting with $r_0 = f$ and 
\begin{align}
\label{eqn:popboost}
r_{b+1} = r_b - \lambda \sum_{i=1}^{m(p+1)} \frac{\inner{r_b}{g_{b,i}}}{\norm{h_{b,i}}^2} g_{b,i},
\end{align}
with the same boosted trees used. \nblue{They concluded that these processes converge to the consistent estimate in the completion of the decision tree family $\mathcal{T}$. As a result, we can achieve asymptotic consistency as long as the gap between the sample process and this population process diminishes fast enough along with the increase of sample size.}
\subsection{Consistency} 
Lemma \ref{lemma:donsker} helps to quantify the discrepancy between tree boosted VCM fits and their population versions conditioned on the sequence of trees used during boosting by decomposing a decision tree having terminal nodes in $\mathcal{R}$ into several hyper rectangles. This strategy also applies to tree boosted VCM. To further achieve consistency, we pose several additional conditions. 
\begin{enumerate}[(C1)]
\item In practice we require the learning rate $\lambda$ to satisfy that $\lambda \leq (1+p)^{-1}$, while in proofs we use  $\lambda = (1+p)^{-1}$.
\item All terminal nodes of the trees in the ensemble should have at least $N_n$ observations such that
$N_n \geq O\left(n^{\frac{3}{4}+\eta}\right)$ for some small $\eta > 0$, in which case we will have 
$$
\frac{1}{n}\norm{h_R}^2_n = \frac{1}{n} \sum_{i=1}^n \I{z_i \in R} \geq O\left(n^{-\frac{1}{4}+\eta}\right) 
$$
for all $R \in \mathcal{R}$ that appear as terminal nodes in the ensemble. 
\item We apply early stopping, allowing at most $B = B(n) = o(\log n)$ iterations during boosting. 
\item From the optimization perspective, we also require that trees in the ensemble have terminal nodes that effectively reduce the empirical risk. Consider the best functional rectangular fit during the $b$-th population iteration 
$$
g^* = \arg\max_{g \in \mathcal{G}} \frac{|\inner{r_b}{g}|}{\norm{g}^2}.
$$ 
We expect to empirically select at least one $(R^*, j)$ pair during the iteration to approximate $g^*$ such that 
$$
\frac{|\inner{r_b}{g_{R^*, j}}|}{\norm{g_{R^*,j}}^2} > \nu \cdot  \frac{|\inner{r_b}{g^*}|}{\norm{g^*}^2},
$$ 
for some $0 < \nu < 1$. \nblue{Lemma \ref{lemma:donsker} indicates that by choosing the sample version optimum 
$$
\hat{g}^* = \arg\max_{g \in \mathcal{G}} \frac{|\inner{r_b}{g}_n|}{\norm{g}_n^2},
$$
the above requirement can be hold true in probability for a fixed number of iterations. 
}
\item $\norm{f}^2 = M \leq \infty$. In addition, due to the linear models in the VCM, to achieve consistency we require that $f \in \mbox{span}(\mathcal{G})$. 
\item We also require the identifiability of linear models such that the distribution of $X$ conditioned on any choice of $Z=z$ should spread uniformly, i.e. 
$$
\inf_{R \in \mathcal{R}, j=0,\dots, p} \frac{\norm{g_{R, j}}}{\norm{h_R}} = \alpha_0 > 0.
$$
\item A stronger version of (C6) is to assume the existence of $s>0, c>0$ s.t. $\forall z$ a.e., there exists an open ball $B_z(x_0, s) \in [-1,1]^p$ centered at $x_0 = x_0(z)$ inside of which $P(X=(1,x)|Z=z)$ is bounded below by $c$. In other words, conditioned on any choice of $Z=z$ there is enough spreading sample points in an open region of $X$ that assures model identifiability. 
\end{enumerate}

\nblue{
During local gradient descent, unwanted behaviors can take place when there is local dependent relation between $X$ and $Z$ in the vicinity of some $Z=z$. Extreme cases include $P(X^1=X^2|Z=z) = 1$, two covariates being collinear, or $P(X^1=x | Z=z) = 1$, some covariates having degenerate conditional distributions. These cases prevent the local parametric model from being identifiable, and the introduction of (C6) and (C7) avoids those cases.
}

\begin{theorem} Under conditions (C1)-(C5), consider function $f \in \mbox{span}(\mathcal{G})$,
\label{thm:consistency}
$$
\Ex{(x^*,z^*)}{|\hat{\beta}_B(z^*)^Tx^* - f(x^*,z^*)|^2} = o_p(1), n \to \infty,
$$
for making predictions at a random point $(x^*,z^*)$ which are independent from but identically distributed as the training data. 
\end{theorem}
\begin{corollary} If we further assume (C6),
\label{cor:consistency}
$$
\hat{\beta}_B(z^*) \xlongrightarrow{p} \beta(z^*), n \to \infty.
$$
\end{corollary}

Corollary \ref{cor:consistency} justifies the varying coefficient mappings as valid estimators for the true varying linear relationship. Although we have not explicitly introduced any continuity condition on $\beta$, it is worth noticing that (C5) requires $\beta$ to have relatively invariant local behavior. Although one region in $\mathbb{A}$ of any size can be eventually detected by the growing $n$ to fit into a terminal node with sufficient sample points required by (C2), such rate is too loose to guarantee the detection of a small area with a small sample. As a result, tree boosted VCM should be the most ideal when $\mathbb{A}$ is heterogeneous with a few big and flat regions. When we consider the interpretability of tree boosted VCM, consistency is also the sufficient theoretical guarantee for \textit{local fidelity} discussed in \citet{ribeiro2016should} that an interpretable local method should also yield accurate local relation between covariates and responses. 

\section{More Empirical Study} 
\subsection{Identifying Signals} \nblue{Our theory suggests that tree boosted VCM is capable of identifying local linear structures and their coefficients accurately. To demonstrate this in practice, }we apply it to the following regression problem with higher order feature interaction on the action space. \begin{gather*}
z = (z^1, z^2, z^3, z^4), \quad z^1, z^2 \sim \mbox{Unif}\{1,\dots, 10\}, \quad z^3, z^4 \sim \mbox{Unif}[0,1].\\
x \in \mathbb{R}^{7}, \quad x \sim N(0, I_{7}), \quad \epsilon \sim N(0, 0.25).
\end{gather*}
The data generating process is describe by the following pseudo code.
\begin{align*}
&\mbox{if } z^1 < 4: & y =  1 + 3 x^1 + 7 x^2 \\
&\mbox{else if } z^1 > 8: & y = -5 + 2x^1 + 4x^2 + 6x^3\\
&\mbox{else if } z^2 = 1, 3 \mbox{ or } 5: & y = 5 + 5x^2 + 5x^3\\
&\mbox{else if } z^3 < 0.5: & y = 10 + 10x^4 \\
&\mbox{else if } z^4 < 0.4: & y = 10+10x^5 \\
&\mbox{else if } z^3 < z^4: & y = 5-5x^2 - 10x^3 \\
&\mbox{else}: & y = -10 x^1 + 10x^3 
\end{align*}

We utilize a sample of size $10,000$ and use 100 trees of maximal depth of 6 for boosting with constant learning rate of $0.2$. Figure \ref{fig:simulated_highorder} plots the fitted distribution of each coefficient in red against the ground truth in grey, \nblue{with reported MSE 3.28. We observe that all peaks and their intensities properly reflect the coefficient distributions on the action space. Despite the linear expressions, we have tested interaction among all four action covariates of a tree depth of 6 and have not yet achieved convergence, which we conclude as the reasons for large MSE. It manifests the effectiveness of our straightforward implementation of decision trees segmenting the action space.}
\begin{figure}[h] 
   \centering
   \includegraphics[width=0.24\textwidth]{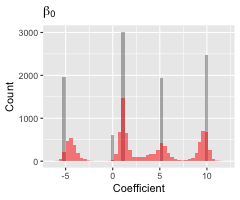} 
   \includegraphics[width=0.24\textwidth]{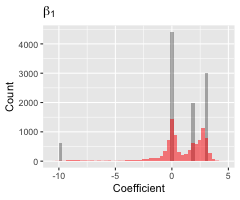}
   \includegraphics[width=0.24\textwidth]{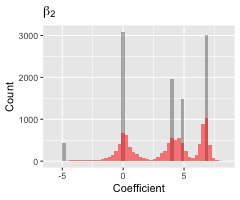}
   \includegraphics[width=0.24\textwidth]{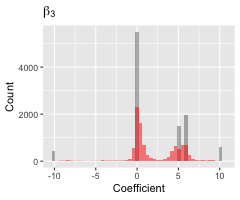}\\
   \includegraphics[width=0.24\textwidth]{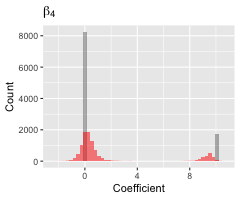}
   \includegraphics[width=0.24\textwidth]{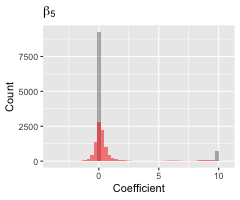} 
   \includegraphics[width=0.24\textwidth]{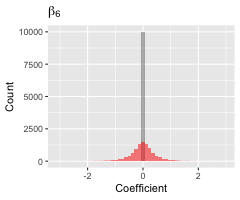} 
   \includegraphics[width=0.24\textwidth]{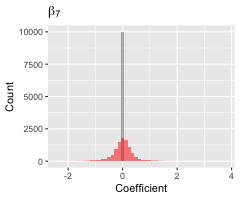} 
   \caption{Histograms of distributions of fitted coefficient values. Color code: ground truth (grey) and tree boosted VCM (red).}
   \label{fig:simulated_highorder}
\end{figure}

\subsection{Model Accuracy}
\nblue{To show the accuracy of our proposed methods, we have selected 12 real world datasets and run tree boosted VCM (marked as TVCM) against other benchmark methods. Table \ref{tab:ce} demonstrates the results under classification settings with three benchmarks: GLM as logistic regression, GLM(S) as a partially saturated logistic regression model where each combination of discrete action covariates acts as fixed effect with its own level, and AdaBoost. Table \ref{tab:l2} demonstrates the results under regression settings. The three benchmarks we choose here are: LM as linear model, LM(S) as a partially saturated linear model, and GBM as the gradient boosted trees. Although with additional structural assumptions, tree boosted VCM performs nearly on a par with both GBM and AdaBoost. It benefits from its capability of modeling the action space without structural conditions to outperform the fixed effect linear model in certain cases. 

\begin{table}[htbp]
   \centering
   \begin{tabular}{cccccc} 
      \toprule
      NAME & GLM & GLM(S) & ADABOOST & TVCM \\
      \midrule
MAGIC04  & 0.208(0.007) & 0.209(0.0065) & 0.13(0.0076) & 0.209(0.0072) \\ 
BANK  & 0.111(0.0044) & 0.1(0.0044) & 0.098(0.0035) & 0.114(0.0043) \\ 
OCCUPANCY  & 0.014(0.0042) & 0.0129(0.0034) & 0.00567(0.0016) & 0.0126(0.0048) \\ 
SPAMBASE  & 0.0749(0.011) & 0.0732(0.012) & 0.0564(0.0098) & 0.0616(0.0097)  \\ 
ADULT  & 0.188(0.004) & 0.155(0.0046) & 0.136(0.0049) & 0.154(0.0028)  \\ 
EGRIDSTAB  & 0.289(0.014) & 0.227(0.018) & 0.179(0.009) & 0.177(0.015) \\ 
       \bottomrule
   \end{tabular}
  \caption{Prediction accuracy of classification and 0-1 loss for six UCI data sets through tenfold cross validation. Results are shown as mean(sd). Sources of some datasets are: BANK\citep{moro2014data} and OCCUPANCY\citep{candanedo2016accurate}.}
   \label{tab:ce}
\end{table}
}

\begin{table}[htbp]
   \centering
   \begin{tabular}{ccccc} 
      \toprule
      NAME & LM & LM(S) & GBM & TVCM \\
      \midrule
BEIJINGPM  & 6478(227) & 5041(203) & 3465(176) & 3942(178) \\ 
BIKEHOUR  & 24590(1630) & 12190(818) & 5791(419) & 6596(597) \\ 
STARCRAFT  & 1.135(0.0622) & 1.116(0.0645) & 1.045(0.0594) & 1.161(0.0596) \\ 
ONLINENEWS  & 0.8544(0.0331) & 0.8377(0.0328) & 0.7826(0.0298) & 0.8183(0.0337) \\ 
ENERGY  & 18.01(4.42) & 9.801(2.16) & 0.5633(0.162) & 9.864(2.27) \\ 
EGRIDSTAB  &  1.01e-03(4.5e-05) & 6.92e-04(3e-05) & 4.31e-04(1.4e-05) & 4.27e-04(8.3e-06) \\ 
       \bottomrule
   \end{tabular}
   \caption{Prediction accuracy of regression and mean square error for six UCI data sets through tenfold cross validation. Results are shown as mean(sd). Sources of some datasets are: BEIJINGPM\citep{liang2015assessing}, BIKEHOUR\citep{fanaee2014event}, ONLINENEWS\citep{fernandes2015proactive} and ENERGY\citep{tsanas2012accurate}.}
   \label{tab:l2}
\end{table}

\subsection{Visual Interpretability: Beijing Housing Price} Here we show the results of applying tree boosted VCM on the Beijing housing data \citep{Beijing}. We take the housing unit price as the target regressed on covariates of location, floor, number of living rooms and bathrooms, whether the unit has an elevator and whether the unit has been refurbished. Specially, location has been treated as the action space represented in pairs of longitude and latitude. Location specific linear coefficients of other covariates are displayed in Figure \ref{fig:beijinghousing}. We allow 200 trees of depth of 5 in the ensemble with a constant learning rate of 0.05. 

\begin{figure}[htbp] 
   \centering
	\includegraphics[width=0.325\linewidth]{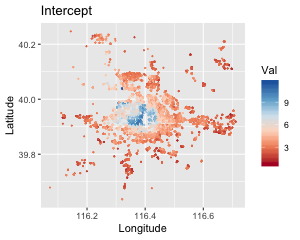} 
	\includegraphics[width=0.325\linewidth]{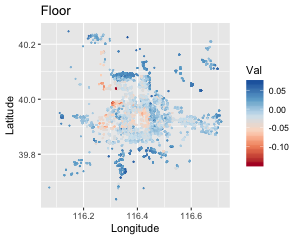} 
	\includegraphics[width=0.325\linewidth]{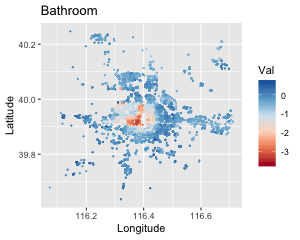} \\ 
	\includegraphics[width=0.325\linewidth]{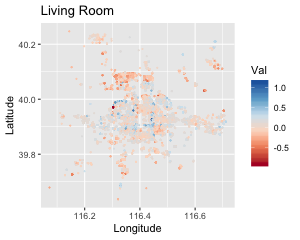}  
	\includegraphics[width=0.325\linewidth]{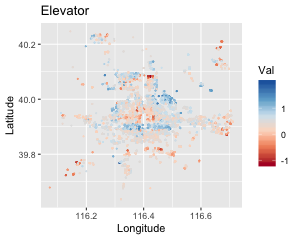} 
	\includegraphics[width=0.325\linewidth]{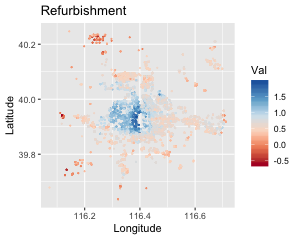} 
   \caption{Beijing housing unit price broken down on several factors.}
   \label{fig:beijinghousing}
\end{figure}
The urban landscape of Beijing is pictured by its old inner circle with a low skyline gradually transitioning to its modern outskirt rim of skyscrapers housing a young and new workforce. Our model intercept provides the baseline of the unit housing prices in each area. Despite their high values, most buildings inside the inner circle are old and not suitable for replanning, so elevators and number of bathrooms are of low contribution to the final price, while their refurbishment gets more attention. In contrast, outskirt housing gains more value if the unit has a complementary elevator and is on higher floor.  

Figure \ref{fig:beijinghousing} provides clear visualization of the fitted tree boosted VCM. Usually these irregular patterns are more likely to be outputs of nonparametric models, while behind each point on our plot is a location-specific linear model predicting the housing price breaking down to different factors.  

\subsection{Fitting Other Model Class} \nblue{As mentioned, since VCM is the generalization of many specific models, our proposed fitting algorithm and analysis should apply to them as well. We take partially linear models as an example and consider the following data set from Cornell Lab of Ornithology consisting of the recorded observations of four species of vireos along with the location and surrounding terran types. We apply a tree boosted VCM under logistic regression setting using the longitude and the latitude as the action space and all rest covariates as linear effects, obtaining the model demonstrated by Table \ref{tab:plinear}. The intercept plot suggests the trend of observed vireos favoring cold climate and inland environment, while the slopes of different territory types indicate a strong preference towards the low elevation between deciduous forests and evergreen needles.

\begin{table}[h]
   \begin{minipage}[]{0.32\textwidth}
   \includegraphics[width=\textwidth]{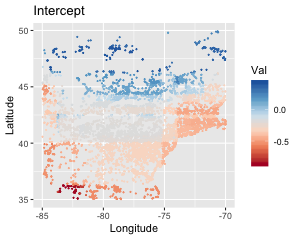}  
   \label{fig:orni}
   \end{minipage}
   \begin{small}
   \begin{minipage}[]{0.66\textwidth}
   \begin{tabular}{@{} cccc@{}} 
      \toprule
      Covariate & Slope & Covariate & Slope \\
      \midrule 
      Elevation & 9.65e-04 & Barren& -1.59e+03 \\
      Shallow Ocean & -1.88e+03 &  Evergreen Broad & 2.77e+02  \\
      CoastShore Lines &  -6.51e+01 & Deciduous Needle &   2.57e+02 \\
      Shallow Inland &   9.39e+01 &  Deciduous Broad &   2.72e+02 \\
      Moderate Ocean&  -1.18e+03  & Mixed Forest &  7.32e+01\\
      Deep Ocean &   -5.12e+03 & Closed Shrubland &   -1.19e+03 \\
      Evergreen Needle &   -4.54e+02 &  Open Shrubland &  8.60e+01 \\
      Grasslands&   -4.49e+02  & Woody Savannas &  -5.75e+02\\
      Croplands & -4.29e+02 & Savannas & -7.46e+02   \\ 
      Urban Built & -6.62e+02 & & \\
      \bottomrule
   \end{tabular}
 \end{minipage}
 \end{small}
 \label{tab:plinear}
 \caption{Fitting a partially linear model using tree boosted VCM. Plot on the left shows the nonparametric intercept. Table on the right shows the coefficients of predictive covariates.}
\end{table}
}

\section{Shrinkage, Selection and Serialization}
Tree boosted VCM is compatible with any alternative boosting strategy in place of the boosting steps (B3) and (B4), such as the use of subsampled trees \citep{friedman2002stochastic, zhou2018boulevard}, univariate or bivariate trees \citep{lou2012intelligible, hothorn2013mboost} or adaptive shrinkage (dropout) \citep{rashmi2015dart, rogozhnikov2017infiniteboost, zhou2018boulevard}. While these alternative approaches have been empirically shown to help avoid overfitting or provide more model interpretability, we also anticipate that the corresponding varying coefficient mappings would inherit certain theoretical properties. For instance, \citet{zhou2018boulevard} introduced a regularized boosting framework named Boulevard that guarantees finite sample convergence and asymptotic normality of its predictions. Incorporating Boulevard into our tree boosted VCM framework requires the changes to (B3) and (B4) such that
\begin{enumerate}
\item [(B3*)] For each $j$, find a good fit $t_{b+1}^j \in \mathcal{T}: \mathbb{A} \to \mathbb{R}$ on $(z_i, \Delta^j_{\beta_i})$ for $i \in w \subset \{1,\dots,n\}$ a random subsample.
\item [(B4*)] Update $\hat{\beta}_b$ with learning rate $\lambda < 1$.
\begin{gather*}
\hat{\beta}_{b+1}(\cdot) = \frac{b}{b+1}\hat{\beta}_{b}(\cdot) + \frac{\lambda}{b+1} \Gamma_M\left(\begin{bmatrix}t_{b+1}^0(\cdot) \\ \vdots \\ t_{b+1}^p(\cdot)\end{bmatrix}\right),
\end{gather*}
where $\Gamma_M$ truncates the absolute value at some $M > 0$.
\end{enumerate}
By taking the same approach in the original paper, we can show that boosting VCM with Boulevard will also yield finite sample convergence to a fixed point. 

\nblue{Boulevard modifies the standard boosting strategy to the extent that new theoretical results have to be developed specifically. In contrast, there are other boosting variations that fall directly under the theoretical umbrella of tree boosted VCM.} Our discussion so far assumes we run boosting iterations with a distinct tree built for each coefficient component while these trees are simultaneously constructed using the same batch of pseudo-residuals. Despite the possibility to utilize a single decision tree with multidimensional response to produce all components, as long as we build separate trees sequentially, the question arises that whether we should update the pseudo-residuals on the fly.

One advantage of doing so is the minimized boosting iteration from $(1+p)$ trees down to one tree, allowing us to use much larger learning rate $\lambda \leq 1/2$ instead of $\lambda \leq (1+p)^{-1}$ without changing the arguments we used to establish the consistency. We also anticipate that doing so in practice moderately reduces the cost as the gradients become more accurate for each tree. Here we will consider two approaches to conduct the on-the-fly updates.

In \citet{hothorn2013mboost} the authors proposed the component-wise linear least squares for boosting where they select which $\beta$ to update using the stepwise optimal strategy, i.e., choose $j_b$ and update $\beta^{j_b}$ if
$$
j_b = \arg\min_{j=0,\dots,p} \sum_{i=1}^n l(y_i, x_i^T(\hat{\beta}_b + \lambda t_b^{j}e_j)(z_i)),
$$
the component tree that reduces the empirical risk the most. As a result, (B4) in Algorithm now updates 
\begin{gather*}
\hat{\beta}_{b+1} = \hat{\beta}_{b} + \lambda t_{b+1}^{j_b} e_{j_b}.
\end{gather*}
Notice that finding this optimum still requires the comparison among all components, therefore does not save any training cost when there are no better means or prior knowledge to help detect which component stands out. \nblue{That being said,} the optimal move is compatible with the key condition (C4) we posed to ensure consistency. Namely, it still guarantees that the population counterpart of boosting is efficient in reducing the gap between the estimate and the truth.  However, this greedy strategy also complicates the pattern of the sequence in which $\beta$'s get updated. 

 \nblue{Serialization refers to the cases when the $\beta$'s are being updated in some predetermined order. A similar model is covered by \citet{lou2012intelligible, lou2013accurate} where the authors applied univariate generalized additive models (GAM) to perform model distillation, which was refined in \citet{tan2017distill} using decision trees. Their models can either be built through backfitting which eventually produces one additive component for each covariate, or through boosting that generates a sequence of additive terms. 
 
Applying the rotation of coordinates to tree boosted VCM, we can break each of the original boosting iterations into (p+1) micro steps to write
$$
\hat{\beta}_{b, j} = \hat{\beta}_{b, j-1} - \lambda \nabla_{\beta^j} l(y, x^T\beta) \Big|_{\beta = \hat{\beta}_{b, j-1} (z)},
$$
with j rotating through $0,\dots, p$. This procedure immediately updates the pseudo-residuals after each component tree is built. There are two feasible approaches if we intend to employ tree boosted VCM to achieve the same univariate GAM model. Either we can place all covariates into the action space and use only univariate decision trees to perform the serialized boosting, or we can directly apply tree boosted VCM to get additive models that are univariate with respect to the predictive covariates.

However, this procedure is not compatible with our consistency conclusion as the serialized boosting fails to guarantee (C4): each micro boosting step on a single coordinate relies on the current pseudo gradients instead of the gradients before the entire rotation. One solution is to consider an alternative to the determined updating sequence by randomly and uniformly proposing the coordinate to boost. In this regard,
$$
\hat{\beta}_{b} = \hat{\beta}_{b} - \lambda \nabla_{\beta^j} l(y, x^T\beta) \Big|_{\beta = \hat{\beta_{b}} (z)},
$$
where $j \sim \mbox{Unif}\{0,\dots, p\}$. This stochastic sequence solves the compatibility issue by satisfying (C4) with a probability bounded from below.
}
\section{Conclusion} In this paper we studied the combination of gradient boosting and varying coefficient models. We introduced the algorithm of fitting tree boosted VCM under generalized linear model setting. We have proved the consistency of its coefficient estimators with an early stopping strategy. We have also demonstrated the performance of this method through diverse types of empirical studies. As a result we are confident to add gradient boosted decision trees into the collection of apt varying coefficient mappings. 

We also discussed the model interpretability of tree boosted VCM, especially its parametric part as a self-explanatory component of reasoning about the covariates. Consistency guarantees the effectiveness of the interpretation. We also demonstrated a few visual aid for explaining tree boosted VCM.

There are a few future directions of this research. The performance of tree boosted VCM depends on boosting scheme as well as the component decision trees involved. It is worth analyzing the discrepancies between different approaches of building the boosting tree ensemble. Also, our discussions in the paper so far concentrate on the generic VCM. As mentioned, multiple regression and additive models can be considered as specifications of generic VCM, hence we expect model-specific conclusions and training algorithms.

\bibliographystyle{chicago}
\bibliography{cite}

\appendix
\section{Proofs}

\subsection{Proof to Lemma \ref{lemma:donsker}}
\begin{proof}
$\xi_{n,6}$ is simply CLT.  For the rest, we will conclude the corresponding function classes are $P-$Donsker. The collection of indicators for hyper rectangles $(-\infty, a_1] \times \dots, (-\infty, a_p] \subseteq \mathbb{R}^p$ is Donsker. By taking difference at most $p$ times we get all elements in $\mathcal{H}$, therefore $\mathcal{G}$, the indicators of $\mathcal{R}$, is Donsker. Thus $\xi_{n,1} = \Op{n^{-\frac{1}{2}}}$.

The basis functions $\mathcal{E} = \{1, x^j, j=1,\dots, p\}$ is Donsker since all elements are monotonic and bounded since $x\in [-1,1]^p$. So $\mathcal{G} = \mathcal{H} \times \mathcal{E}$ is Donsker, which gives
$\xi_{n,2} = \Op{n^{-\frac{1}{2}}}$ and $\xi_{n,4} = \Op{n^{-\frac{1}{2}}}$.

In addition, for fixed $f$, $f\mathcal{G}$ is therefore Donsker, which gives $\xi_{n,3} = \Op{n^{-\frac{1}{2}}}$. And $\mathcal{G} \times \mathcal{G}$ is Donsker, which gives $\xi_{n,5} = \Op{n^{-\frac{1}{2}}}$.
 
\end{proof} 

\subsection{Proof to Theorem \ref{thm:consistency}}
To supplement our discussion of norms, it is immediate that $\norm{g_{R,j}} \leq \norm{h_R} \leq 1.$ Another key relation is $\norm{g_{R,j}}_{n,1} \leq \norm{h_R}_{n,1} = \norm{h_R}_n^2.$ We also assume that all $R$'s satisfy the terminal node condition.
\begin{lemma} $\norm{\hat{r}_{b+1}}_n \leq \norm{\hat{r}_{b}}_n$, $\norm{r_{b+1}} \leq \norm{r_{b}}$.
\end{lemma}
\begin{proof}
Consider the $p+1$ trees used for one boosting iteration with the terminal nodes denoted as $R^j_i, 0=1,\dots,p, i=1,\dots, m$, write $I_{R} = \I{z \in R}$.
\begin{align*}
\norm{\hat{r}_{b+1}}_n & = \norm{\hat{r}_b - \lambda\sum_{j=0}^p\sum_{i=1}^m \frac{n^{-1}\inner{\hat{r}_b}{g_{R^j_i, j}}_n}{n^{-1}\norm{h_{R^j_i}}_n^2}g_{R^j_i,j}}_n\\
&\leq \sum_{j=0}^p\norm{\sum_{i=1}^m \left(\lambda\hat{r}_{b}I_{R^j_i}- \frac{n^{-1}\inner{\lambda\hat{r}_bI_{R^j_i}}{g_{R^j_i, j}}_n}{n^{-1}\norm{h_{R^j_i}}_n^2}g_{R^j_i,j}\right)}_n \\
&= \sum_{j=0}^p
\left(\norm{\sum_{i=1}^m \left(\lambda\hat{r}_{b}I_{R^j_i}- \frac{n^{-1}\inner{\lambda\hat{r}_bI_{R^j_i}}{g_{R^j_i, j}}_n}{n^{-1}\norm{h_{R^j_i}}_n^2}g_{R^j_i,j}\right)}_n^2
\right)^{\frac{1}{2}} \\
&= \sum_{j=0}^p
\left(\sum_{i=1}^m \norm{ \lambda\hat{r}_{b}I_{R^j_i}- \frac{n^{-1}\inner{\lambda\hat{r}_bI_{R^j_i}}{g_{R^j_i, j}}_n}{n^{-1}\norm{h_{R^j_i}}_n^2}g_{R^j_i,j}}_n^2
\right)^{\frac{1}{2}} \\
&\leq \sum_{j=0}^p
\left(\sum_{i=1}^m \norm{ \lambda\hat{r}_{b}I_{R^j_i}}_n^2
\right)^{\frac{1}{2}} = \sum_{j=0}^p \norm{ \lambda\hat{r}_{b}}_n = \norm{\hat{r}_b}_n,
\end{align*}
given $\norm{h_{R^j_i}}_n^2 \geq \norm{g_{R^j_i, j}}_n^2$. Same argument can be applied to the population version hence we get the second part. 

\end{proof}
\begin{lemma} For any $b\leq 0$, as defined in (\ref{eqn:popboost}),  
\label{lemma:bound}
$$
\sup_{x, z} |r_{b+1}(x,z)| \leq 2  \sup_{x, z} |r_{b}(x,z)|.
$$
\end{lemma}
\begin{proof}
As implied by (\ref{eqn:improve}), for $(x,z)$ such that $z \in R$, 
$$
r_{b+1}(x,z) = r_b(x,z) - \lambda \sum_{i=0}^p \frac{\inner{r_b}{g_{R,j}}}{\norm{h_{R}}^2} g_{R, j}(x,z).
$$
The key observation is that
$$
\left|\frac{\inner{r_b}{g_{R,i}}}{\norm{h_{R}}^2}\right| = \left|\frac{\int r_b \I{z \in R} x^j dP}{\int \I{z \in R}^2 dP}\right| \leq \sup_{x,z}|r_b(x,z)|.
$$
Therefore, provided $|g_{R,i}| \leq 1$ and write $z \in R_z$, 
\begin{align*}
\sup_{x,z} |r_{b+1}| &\leq \sup_{x,z} |r_b| + \lambda \sup_{x,z} \sum_{i=0}^{p} \left|\frac{\inner{r_b}{g_{R_z,i}}}{\norm{h_{R_z}}^2}\right|\left| g_{R_z,i}\right| \\
& \leq \sup |r_b| + \lambda \sum_{i=0}^{p}\sup |r_b| \cdot 1 \\
& = 2 \sup_{x, z} |r_b|.
\end{align*}  
Recursively we can conclude that $\sup_{x,z} |r_b| \leq 2^b \sup_{x,z} |r_0|$. \end{proof}

\begin{lemma} Under conditions (C1)-(C6), 
$$
\norm{\hat{r}_B}^2 = \norm{r_B}^2 + \sigma_{\epsilon}^2 + \op{1},
$$
where $\sigma^2_{\epsilon} = \norm{\epsilon}^2$.
\end{lemma}
\begin{proof}
Recall that
$$
\hat{r}_{b+1} = 
\hat{r}_b -  \lambda\sum_{j=0}^p\sum_{i=1}^m \frac{n^{-1}\inner{\hat{r}_b}{g_{R^j_i, j}}_n}{n^{-1}\norm{h_{R^j_i}}_n^2}g_{R^j_i,j}
=\hat{r}_b - \lambda \sum_{i=1}^{m(p+1)} \frac{n^{-1}\inner{\hat{r}_b}{g_{b, i}}_n}{n^{-1}\norm{h_{b,i}}_n^2} g_{b,i},
$$
and 
$$
r_{b+1}
=r_b -  \lambda\sum_{j=0}^p\sum_{i=1}^m \frac{\inner{r_b}{g_{R^j_i, j}}}{\norm{h_{R^j_i}}^2}g_{R^j_i,j} 
= r_b - \lambda  \sum_{i=1}^{m(p+1)} \frac{\inner{r_b}{g_{b, i}}}{\norm{h_{b,i}}^2} g_{b,i}.
$$
Therefore
\begin{align*}
\hat{r}_{b+1} - r_{b+1} & = (\hat{r}_b - r_b) + \lambda \sum_{j=0}^p\sum_{i=1}^m 
\left(
\frac{\inner{r_b}{g_{R^j_i, j}}}{\norm{h_{R^j_i}}^2} - \frac{n^{-1}\inner{\hat{r}_b}{g_{R^j_i, j}}_n}{n^{-1}\norm{h_{R^j_i}}_n^2}
\right) g_{R^j_i, j} \\
&= (\hat{r}_b - r_b) + \lambda \sum_{i=1}^{m(p+1)} \left(\frac{\inner{r_b}{g_{b,i}}}{\norm{h_{b,i}}^2} - \frac{n^{-1}\inner{\hat{r}_b}{g_{b,i}}_n}{n^{-1}\norm{h_{b,i}}_n^2}\right) g_{b,i}\\
& \triangleq (\hat{r}_b - r_b) + \lambda \delta_{b} \\
& =  (\hat{r}_0 - r_0)+\lambda \sum_{j=0}^b \delta_j = \epsilon + \lambda\sum_{j=0}^b \delta_j.
\end{align*}
Since for each fixed $j$, all $R_i^j$ are disjoint, we therefore define that
$$
\gamma_b = \sum_{j=0}^{p} 
\sup_{i=1,\dots, m} \left|
\frac{\inner{r_b}{g_{R^j_i, j}}}{\norm{h_{R^j_i}}^2} - \frac{n^{-1}\inner{\hat{r}_b}{g_{R^j_i, j}}_n}{n^{-1}\norm{h_{R^j_i}}_n^2}
\right|,
$$
which guarantees $\sup_{x,z} |\delta_b| \leq \gamma_b$. To bound $\gamma_b$, without loss of generality, we consider a single term involved such that
\begin{align*}
\frac{\inner{r_b}{g_{b}}}{\norm{h_{b}}^2} - \frac{n^{-1}\inner{\hat{r}_b}{g_{b}}_n}{n^{-1}\norm{h_{b}}_n^2} & \triangleq \left(\frac{u}{v} - \frac{\hat{u}}{\hat{v}}\right)\\
&= \left(\frac{u-\hat{u}}{v} + \left(\frac{1}{v}-\frac{1}{\hat{v}}\right)\hat{u}\right).
\end{align*}
First consider 
\begin{align*}
\hat{u} - u & = \frac{1}{n} \inner{\hat{r}_b}{g_{b}}_n - \inner{r_b}{g_{b}} \\
& = \frac{1}{n} \inner{\epsilon + r_b + \sum_{j=0}^{b-1}\delta_j}{g_b}_{n}  -  \inner{r_b}{g_{b}} \\
& = \frac{1}{n} \inner{\epsilon}{g_b}_n + \left( \frac{1}{n}\inner{r_b}{g_b}_n - \inner{r_b}{g_b}\right) + \left(\sum_{j=0}^{b-1} \frac{1}{n}\inner{\delta_j}{g_b}_n\right).
\end{align*}
Per Lemma \ref{lemma:bound}, we have
$$
\left|\frac{1}{n} \inner{\epsilon}{g_b}_n\right| \leq \xi_n
$$
and, by iteratively applying Lemma \ref{lemma:bound} and setting $C_0 = \max(\sup_{x,z} |f|, 1),$ 
\begin{align*}
&\left| \frac{1}{n}\inner{r_b}{g_b}_n - \inner{r_b}{g_b}\right| \\
=& \left| \frac{1}{n}\inner{f - \lambda \sum_{j=0}^{b-1}\sum_{i=1}^{m(p+1)}\frac{\inner{r_j}{g_{j, i}}}{\norm{h_{j,i}}^2} g_{j,i}}{g_b}_n - \inner{f - \lambda \sum_{j=0}^{b-1}\sum_{i=1}^{m(p+1)}\frac{\inner{r_j}{g_{j, i}}}{\norm{h_{j,i}}^2} g_{j,i}}{g_b}\right|\\
\leq & \left|\frac{1}{n}\inner{f}{g_b}_n - \inner{f}{g_b}\right| + \lambda \sum_{j=0}^{b-1}\sum_{i=1}^{m(p+1)} \left| \frac{1}{n}\inner{\frac{\inner{r_j}{g_{j, i}}}{\norm{h_{j,i}}^2} g_{j,i}}{g_b}_n - \inner{\frac{\inner{r_j}{g_{j, i}}}{\norm{h_{j,i}}^2} g_{j,i}}{g_b}\right| \\
\leq & \left|\frac{1}{n}\inner{f}{g_b}_n - \inner{f}{g_b}\right| + \lambda \sum_{j=0}^{b-1}\sum_{i=1}^{m(p+1)} \left|\frac{\inner{r_j}{g_{j, i}}}{\norm{h_{j,i}}^2}\right|\left| \frac{1}{n}\inner{g_{j,i}}{g_b}_n - \inner{ g_{j,i}}{g_b}\right| \\
\leq & \xi_n + \lambda \sum_{j=0}^{b-1} \sup |r_j| m(p+1) \xi_n \\
\leq & \xi_n + C_0\sum_{j=0}^{b-1} 2^j m \xi_n \\
\leq & C_02^bm\xi_n.
\end{align*}
The last term could be bounded by
\begin{align*}
\left|\sum_{j=0}^{b-1} \frac{1}{n}\inner{\delta_j}{g_b}_n\right| &\leq \frac{1}{n} \sum_{j=0}^{b-1} \norm{\delta_j}_{n,\infty} \norm{g_b}_{n,1} \leq \frac{1}{n}\sum_{j=0}^{b-1} \gamma_j \norm{h_b}_{n}^2, 
\end{align*}
where
$$
\norm{g_b}_{n,1} = \sum_{i=1}^n |g_b(x_i)|,\quad\norm{\delta_j}_{n,\infty} = \sup_{i=1,\dots, n} |\delta_j(x_i)|.
$$
Hence 
$$
|\hat{u} - u| \leq C_02^bm\xi_n + \frac{1}{n}\sum_{j=0}^{b-1} \gamma_j \norm{h_b}_{n}^2.
$$
In order to bound $|\hat{u}|$, we notice
\begin{align*}
|\hat{u}|  = \left|\frac{1}{n}\inner{\hat{r}_b}{g_{b}}_n\right| 
&\leq \left(\frac{1}{n}\norm{\hat{r}_n}_n^2\right)^{\frac{1}{2}}\cdot \left(\frac{1}{n}\norm{g_b}_n^2\right)^{\frac{1}{2}}\\
&\leq \left(\frac{1}{n}\norm{\hat{r}_0}_n^2\right)^{\frac{1}{2}}\cdot \left(\frac{1}{n}\norm{g_b}_n^2\right)^{\frac{1}{2}}\\
& = \left(\frac{1}{n}\norm{f+\epsilon}_n^2\right)^{\frac{1}{2}}\cdot \left(\frac{1}{n}\norm{g_b}_n^2\right)^{\frac{1}{2}}\\
& \leq (M+\sigma^2_{\epsilon} + \xi_n) \cdot \norm{g_b} \\
& \leq (M_0 + \xi_n) \cdot \norm{h_b}.
\end{align*}
Therefore, we get an upper bound for 
\begin{align*}
\left|\frac{\inner{r_b}{g_{b}}}{\norm{h_{b}}^2} - \frac{n^{-1}\inner{\hat{r}_b}{g_{b}}_n}{n^{-1}\norm{h_{b}}_n^2} \right| &\leq \frac{|\hat{u}-u|}{|v|} + \left|\frac{1}{v} - \frac{1}{\hat{v}}\right| |\hat{u}| \\
& =\frac{C_02^bm\xi_n + n^{-1}\sum_{j=0}^{b-1} \gamma_j \norm{h_b}_{n}^2}{\norm{h_b}^2} + \frac{\xi_n \cdot (M_0+\xi_n) \cdot \norm{h_b}}{\norm{h_b}^2 \cdot n^{-1}\norm{h_b}_n^2}\\
& \leq \frac{C_02^bm\xi_n}{\norm{h_b}^2} + \sum_{j=0}^{b-1} \gamma_j \left(1 + \frac{\xi_n}{\norm{h_b}^2}\right) + \frac{\xi_n(M_0 + \xi_n)}{\norm{h_b}(\norm{h_b}^2-\xi_n)}.
\end{align*}
Denote $h$ be the global minimum of the ensemble that $h = \min_{b,i,j} \norm{h_{R^j_i}}$, since $m \leq (h^2-\xi_n)^{-1}$, we obtain
\begin{align*}
\gamma_b &\leq (p+1) \left(
\frac{C_02^b\xi_n}{h^2(h^2-\xi_n)} + \sum_{j=0}^{b-1} \gamma_j \left(1 + \frac{\xi_n}{h^2}\right) + \frac{\xi_n(M_0 + \xi_n)}{h(h^2-\xi_n)}
\right).
\end{align*}
We would like to mention the elementary result that for a series $\{x_n\}$ satisfying
$$
x_n \leq 2^na + \sum_{i=0}^{n-1}bx_i + c,
$$
the partial sums satisfy
$$
\sum_{i=0}^n x_n \leq a\left(\frac{1-\left(\frac{2}{1+b}\right)^{n+1}}{1-\frac{2}{1+b}} \right)(1+b)^n - \frac{c}{b}.
$$
Hence, we can verify this upper bound that
\begin{align*}
\label{eqn:n1}
\sum_{j=0}^{B-1} \gamma_j &\leq (1+p)^B\left( \frac{C_0}{h^2-\xi_n}\left( 2+\frac{\xi_n}{h^2}\right)^{B}\left(1-\left(1-\frac{\frac{\xi_n}{h^2}}{2+\frac{\xi_n}{h^2}}\right)^{B-1} \right) - \frac{\xi_n(M_0+\xi_n)}{h(h^2-\xi_n)\left(1+\frac{\xi_n}{h^2}\right)}\right)
\end{align*}
Recall the rates that $B = o(\log n), h^2 = \Op{n^{-\frac{1}{4}+\eta}}, \xi_n = \Op{n^{-\frac{1}{2}}}$, thus
\begin{gather*}
\left(2+\frac{\xi_n}{h^2}\right)^B  = 2^B \cdot \Op1, \quad 
1-\left(1-\frac{\frac{\xi_n}{h^2}}{2+\frac{\xi_n}{h^2}}\right)^{B-1} = \frac{\xi_n}{h^2} \cdot \Op1, \\
\frac{\xi_n(M_0+\xi_n)}{h(h^2-\xi_n)\left(1+\frac{\xi_n}{h^2}\right)} = \frac{\xi_n}{h^3}\cdot \Op1.
\end{gather*}
Hence,
\begin{align*}
\sum_{j=0}^{B-1} \gamma_j  
& \leq (1+p)^B \left( \frac{ C_0}{h^2} \cdot 2^B \cdot \frac{\xi_n}{h^2} - \frac{\xi_n}{h^3}\right)\Op1 = o_p(1),
\end{align*}
which is equivalent to 
$$
\norm{\sum_{j=0}^{B-1}\delta_j} \leq \sum_{j=0}^{B-1}\norm{\delta_j} \leq \sum_{j=0}^{B-1}\gamma_j = \op{1}.
$$
Combining all above we have
\begin{align*}
\norm{\hat{r}_B}^2 & = \norm{r_B + \epsilon + \lambda \sum_{j=0}^{B-1} \delta{j}}^2 \\
& \leq \norm{\epsilon}^2 + \norm{r_B}^2 + \lambda^2 \norm{\sum_{j=0}^{B-1}\delta_j}^2 + 2 \lambda \norm{r_B+\epsilon} \norm{\sum_{j=0}^{B-1}\delta_j}\\
&=\sigma_{\epsilon}^2 + \norm{r_B}^2  + \op{1}.
\end{align*}
\end{proof}

\begin{lemma}
Under condition (C1)-(C6), for any $\rho>0$ there exists $B_0 = B_0(\rho)$ and $n_0 = n_0(\rho)$ such that for all $n > n_0$,
$$
P(\norm{r_{B_0}} \leq \rho) \geq 1-\rho.
$$
\end{lemma}
\begin{proof}
Lemma 3 in \citet{buhlmann2002consistency} proves this statement for rectangular indicators. By fixing $\lambda = (1+p)^{-1}$ and introducing conditions (C3) and (C4), formula (11) in \citet{buhlmann2002consistency} still holds in terms of the single terminal node in each of the trees that corresponds to our defined $R^*$. Therefore cited Lemma 3 holds for our boosted trees. The conclusion is therefore reached by the assumption that $f \in \mbox{span}(\mathcal{G})$.
\end{proof}

\begin{proof}[Proof to main Theorem] For a given $\rho>0$, since $\hat{r}_B(x^*, z^*) - f(x^*,z^*)$ is independent of $ \epsilon$,  
\begin{align*}
\Ex{(x^*,z^*)}{|\hat{\beta}_B(z^*)^Tx^* - f(x^*,z^*)|^2} & = \Ex{(x^*,z^*)}{|\hat{r}_B(x^*, z^*) - f(x^*,z^*)- \epsilon|^2} - \norm{\epsilon}^2\\
& = \norm{\hat{r}_B}^2 - \norm{\epsilon}^2 \\
& \leq \norm{r_B}^2 + \op{1} \\
& \leq \norm{r_{B_0}}^2 + \op{1} \\
& \leq \rho \Op1 + \op{1}. 
\end{align*}
We reach the conclusion by sending $\rho \to 0$.
\end{proof}

\subsection{Proof to Corollary \ref{cor:consistency}}
\begin{proof} We prove by contradiction. Assume there exists $0< \epsilon_0 < s, c_0 > 0$ s.t. $$P(\norm{\beta_B(z^*) - \beta(z^*)}^2>\epsilon_0) \geq c_0$$ for any sufficiently large $n$. Fix $n$ and consider any $z_0$ s.t. $\norm{\beta_B(z_0) - \beta(z_0)}>\epsilon_0$. The corresponding open ball $B(x_0, s)$ has volume $v_0 = O(s^p)$. Write $\beta = \bmat{\beta^0\\\beta^{-0}},$ 
\begin{align*}
& \int_{B(x_0,s)} \inner{\bmat{1 \\x}}{\beta_B(z_0) - \beta(z_0)}^2 dP_{x|z_0} \\
\geq & c\int_{B(x_0,s)} \inner{\bmat{1 \\x}}{\beta_B(z_0) - \beta(z_0)}^2 dx \\
\geq & cv_0\inner{\bmat{1 \\x_0}}{\beta_B(z_0) - \beta(z_0)}^2 + c\int_{B(0,s)}\inner{\bmat{1 \\x}}{\beta_B(z_0) - \beta(z_0)}^2 dx\\
\geq & cv_0 (\beta_B(z_0)^0 - \beta(z_0)^0)^2 + ct_0 \norm{\beta_B(z_0)^{-0} - \beta(z_0)^{-0}}^2 \\
\geq & c\min(v_0, t_0) \epsilon_0.
\end{align*}
where $t_0 = \int_{B(0,s)}x_1^2 dx = O(s^p).$
That is equivalent to 
$$
\Ex{(x^*,z^*)}{|\hat{\beta}_B(z^*)^Tx^* - f(x^*,z^*)|^2} > c\min(v_0, t_0) \epsilon_0,
$$
contradicting Theorem \ref{thm:consistency}.
\end{proof}

\end{document}